\documentclass[sigconf]{acmart}

\usepackage{graphicx}
\usepackage{colortbl} 
\usepackage{xcolor}
\usepackage{array}
\usepackage{enumerate}
\usepackage{xspace}
\usepackage{tikz}
\usepackage{soul}
\usepackage{amstext}
\usepackage{amsmath}

\usepackage{amssymb}
\usepackage{amsthm}
\usepackage{algorithm}
\usepackage{algorithmic}
\usepackage{subfigure}
\usepackage{multirow}
\usepackage{makecell}
\usepackage{color}

\newtheorem{proposition}{Proposition}

\AtBeginDocument{
  }

\setcopyright{acmcopyright}
\copyrightyear{2018}
\acmYear{2018}
\acmDOI{XXXXXXX.XXXXXXX}

\acmConference[Conference acronym 'XX]{Make sure to enter the correct
  conference title from your rights confirmation emai}{June 03--05,
  2018}{Woodstock, NY}

\acmPrice{15.00}
\acmISBN{978-1-4503-XXXX-X/18/06}

\begin{document}

\title{Frequency-aware Graph Signal Processing for Collaborative Filtering}

\author{Jiafeng Xia}
\affiliation{
  \institution{Fudan University}
  \city{Shanghai}
  \country{China}
}
\additionalaffiliation{
    \institution{School of Computer Science, Shanghai Key Laboratory of Data Science, Fudan University}
  \city{Shanghai}
  \country{China}
}
\email{jfxia19@fudan.edu.cn}

\author{Dongsheng Li}
\affiliation{
  \institution{Microsoft Research Asia}
  \city{Shanghai}
  \country{China}
 }
 \email{dongsli@microsoft.com}

\author{Hansu Gu}
\affiliation{
  \city{Seattle}
  \country{United States}
}
\email{hansug@acm.org}

\author{Tun Lu}
\authornotemark[1]
\affiliation{
  \institution{Fudan University}
  \city{Shanghai}
  \country{China}
}
\additionalaffiliation{
    \institution{Fudan Institute on Aging, MOE Laboratory for National Development and Intelligent Governance, and Shanghai Institute of Intelligent Electronics \& Systems, Fudan University}
  \city{Shanghai}
  \country{China}
}
\authornote{Corresponding author.}
\email{lutun@fudan.edu.cn}

\author{Peng Zhang}
\authornotemark[1]
\authornotemark[3]
\affiliation{
  \institution{Fudan University}
  \city{Shanghai}
  \country{China}}
\email{zhangpeng_@fudan.edu.cn	}

\author{Li Shang}
\authornotemark[1]
\affiliation{
  \institution{Fudan University}
  \city{Shanghai}
  \country{China}}
\email{lishang@fudan.edu.cn}

\author{Ning Gu}
\authornotemark[1]
\affiliation{
  \institution{Fudan University}
  \city{Shanghai}
  \country{China}}
\email{ninggu@fudan.edu.cn}

\newcommand{\ours}[0]{FaGSP\xspace}

\begin{abstract}
Graph Signal Processing (GSP) based recommendation algorithms have recently attracted lots of attention due to its high efficiency. However, these methods failed to consider the importance of various interactions that reflect unique user/item characteristics and failed to utilize user and item high-order neighborhood information to model user preference, thus leading to  sub-optimal performance. To address the above issues, we propose a frequency-aware graph signal processing method (\ours) for collaborative filtering. Firstly, we design a Cascaded Filter Module, consisting of an ideal high-pass filter and an ideal low-pass filter that work in a successive manner, to capture both unique and common user/item characteristics to more accurately model user preference. Then, we devise a Parallel Filter Module, consisting of two low-pass filters that can easily capture the hierarchy of neighborhood, to fully utilize high-order neighborhood information of users/items for more accurate user preference modeling. Finally, we combine these two modules via a linear model to further improve recommendation accuracy. Extensive experiments on six public datasets demonstrate the superiority of our method from the perspectives of prediction accuracy and training efficiency compared with state-of-the-art GCN-based recommendation methods and GSP-based recommendation methods.
\end{abstract}

\begin{CCSXML}
<ccs2012>
   <concept>
       <concept_id>10002951.10003317.10003347.10003350</concept_id>
       <concept_desc>Information systems~Recommender systems</concept_desc>
       <concept_significance>500</concept_significance>
       </concept>
 </ccs2012>
\end{CCSXML}

\ccsdesc[500]{Information systems~Recommender systems}

\keywords{collaborative filtering, graph signal processing, frequency-aware}

\maketitle

\section{Introduction}

Graph Signal Processing (GSP)~\cite{gsp1, gsp2, gsp3} extends signal processing techniques to graph data, which first transforms signals into frequency domain and then processes specific frequency components to extract and analyze information existed in the signal based on spectral graph theory~\cite{sgt1, sgt2}. Currently, GSP-based recommendation algorithms~\cite{gfcf, pgsp, cfgsp} are receiving increasing attention from researchers due to its parameter-free characteristic. Compared with graph-based collaborative filtering methods, these non-parametric methods do not suffer from a time-consuming training phase and thus are highly efficient, while can achieve comparable or better  performance than deep learning-based methods~\cite{lightgcn, simplex, sgled, ngat4rec, niagcn}.

Generally, user interactions contain rich information that depicts user interests and item characteristics. Through designing and exerting different types of filters on user interactions, we can extract different types of information to  model user preference, thereby improving the accuracy of user future interaction prediction. GF-CF~\cite{gfcf} and PGSP~\cite{pgsp} are two representative GSP-based collaborative filtering methods, both of which adopt an ideal low-pass filter and a linear filter to extract information from user/item common characteristic and user and item first-order neighborhood respectively. Though they have shown promising performance, two limitations restrict their expressivity in modeling user preference.

{The first limitation is that they focus on user/item common characteristics and neglect their unique characteristics.} Common characteristics refer to the shared features among many users/items, such as a broad target audience aged 18-60, 
while the unique characteristics refer to those characteristics that can distinguish a user/item from others, such as a narrow target audience aged 18-25, 
both of which are implicit in user interactions. When they use the ideal low-pass filters to predict user future interactions, merely common characteristics will be utilized for prediction, while unique characteristics will be filtered out. Obviously, in this way, they can only roughly recommend items to users based on common characteristics, making the predictions sub-optimal. { The second limitation is that they fail to fully utilize user and item high-order neighborhood information.} The powerful modeling capacity of GCN~\cite{gcn, gin} mainly benefits from aggregating information from both direct and distant neighbors. By analogy, we can infer that user and item high-order neighborhood information, which is obtained by aggregating information from their neighbors of different distances, can also be crucial to model user preference, so as to extract richer information for interaction prediction. However, linear filters, designed on the user/item co-occurrence relationship, can only extract information from first-order neighborhood, leading to insufficient modeling of user preference and sub-optimal performance. 

In this work, we propose a frequency-aware graph signal processing method (\ours) for collaborative filtering to address the above two limitations. Firstly, we design a Cascaded Filter Module to take both common characteristics and  unique characteristics into consideration for interaction prediction, which first uses an ideal high-pass filter to enhance interaction signal by highlighting those interactions that reflect unique characteristics, then uses an ideal low-pass filter over the enhanced signal to predict future interactions. Secondly, we devise a Parallel Filter Module consisting of two low-pass filters, which are designed to capture user and item high-order neighborhood information respectively for user preference modeling. By adjusting the parameters of these two filters, they are capable of capturing information from neighborhood with any order. Finally, we combine these two modules via a linear model to make \ours able to extract rich information from user interactions for user preference modeling and interaction prediction. We conduct extensive experiments on six real-world datasets, and the experimental results demonstrate the superiority of our method from the perspectives of recommendation accuracy and training efficiency compared with existing GCN/GSP-based collaborative filtering methods. 

\section{Preliminaries}
\subsection{Graph Signal Processing}
\subsubsection{Graph Signal}
Given an undirected and unweighted graph $\mathcal{G}=(\mathcal{N}, \mathcal{E})$, where $\mathcal{N}$ and $\mathcal{E}$ represent the node set and edge set respectively, and $|\mathcal{N}|=n$. The graph structure can be represented as an adjacency matrix $\mathbf{A}\in\mathbb{R}^{n\times n}$, and if there is an edge between node $i$ and node $j$, $\mathbf{A}_{ij}=1$, otherwise, $\mathbf{A}_{ij}=0$. The graph laplacian matrix, defined over $\mathcal{G}$, can be represented as $\mathbf{L}=\mathbf{D}-\mathbf{A}$, where $\mathbf{D}=diag(\mathbf{A}\cdot \mathbf{1})\in \mathbb{R}^{n\times n}$ is the degree matrix, and the corresponding normalized laplacian matrix is $\tilde{\mathbf{L}}=\mathbf{D}^{-\frac{1}{2}}\mathbf{L}\mathbf{D}^{-\frac{1}{2}}$. As the normalized graph laplacian matrix is a real and symmetric matrix, it can be decomposed into ${\tilde{\mathbf{L}}}=\mathbf{U}\pmb\Lambda\mathbf{U}^T$, where $\pmb\Lambda=diag(\lambda_1, \lambda_2, \cdots,\lambda_n)$ and $\mathbf{U}=(\mathbf{u}_1, \mathbf{u}_2,\cdots,\mathbf{u}_n)$ are the eigenvalue matrix and eigenvector matrix respectively, and $0\le\lambda_1\le \lambda_2\le \cdots\le \lambda_n\le 2$. The graph signal can be defined by a mapping $h:\mathcal{N}\rightarrow \mathbb{R}$, and it can be represented as a vector $\mathbf{x}=\{x_1, x_2, \cdots, x_n\}$, where $x_i$ is the signal strength of the $i$-th node. The graph quadratic form is widely used to measure the smoothness of the graph signal, which is defined as 
\begin{equation}
S(\mathbf{x})=\frac{\mathbf{x}^T{\mathbf{L}}\mathbf{x}}{||\mathbf{x}||_2}.\label{eq:sx}
\end{equation}
If $S(\mathbf{x_1})<S(\mathbf{x_2})$, then the graph signal $\mathbf{x}_1$ is smoother than $\mathbf{x}_2$.

\subsubsection{Graph Filter}
GSP uses graph filter $\mathcal{F}$ to analyse graph signal:
\begin{equation}
\mathcal{F} = \mathbf{U}diag(f(\lambda_1), f(\lambda_2), \cdots,f(\lambda_n))\mathbf{U}^T=\mathbf{U}\mathbf{F}\mathbf{U}^T,\label{eq:filter}
\end{equation}
where $\mathbf{U}$ is the eigenvector matrix of normalized laplacian matrix $\tilde{\mathbf{L}}$. $f(\cdot)$ is a frequency response function that determines whether some frequency component $\lambda_i$ is enhanced or attenuated. Different design of $f(\cdot)$ will lead to the different graph filter $\mathcal{F}$, if $f(\cdot)$ is a monotonic decreasing function, then $\mathcal{F}$ is a low-pass filter, and if $f(\cdot)$ is a monotonic increasing function, then $\mathcal{F}$ is a high-pass filter. Low-pass filter and high-pass filter have different effects to the graph signal $\mathbf{x}$, as shown in the Proposition~\ref{prop:prop1}.

\begin{proposition}\label{prop:prop1}
Low-pass filter can smooth the graph signal $\mathbf{x}$, i.e., reduce $S(\mathbf{x})$, while high-pass filter can coarsen the graph signal $\mathbf{x}$, i.e., enhance $S(\mathbf{x})$.
\end{proposition}

\begin{proof}
Since ${\mathbf{L}}$ is a real and symmetric matrix, it is diagonalizable ${\mathbf{L}}=\mathbf{U}\pmb\Lambda\mathbf{U}^T$, thus $S(\mathbf{x})$ can be represented as:
\begin{equation}
\begin{aligned}
S_1(\mathbf{x})&=\frac{\mathbf{x}^T{\mathbf{L}}\mathbf{x}}{\mathbf{x}^T\mathbf{x}}=\frac{\mathbf{x}^T\mathbf{U}\mathbf{\Lambda}\mathbf{U}^T\mathbf{x}}{\mathbf{x}^T\mathbf{U}\mathbf{U}^T\mathbf{x}}=\frac{\mathbf{y}^T\mathbf{\Lambda}\mathbf{y}}{\mathbf{y}^T\mathbf{y}}=\frac{\sum_{i=1}^n\lambda_i y_i^2}{\sum_{i=1}^n y_i^2},
\end{aligned}
\end{equation}
where $\mathbf{U}$ is the orthogonal eigenvector matrix ($\mathbf{U}^T\mathbf{U}=\mathbf{U}\mathbf{U^T}=\mathbf{I}$), and we denote $\mathbf{y} = {\mathbf{U}}^T\mathbf{x}$, and $y_i = \mathbf{u}_i^T\mathbf{x}$. Suppose there is a filter $\mathcal{F}=\mathbf{U}\mathbf{F}\mathbf{U}^T$ exerting on the signal $\mathbf{x}$ and inducing a new graph signal $\mathbf{z}=\mathcal{F}\mathbf{x}$. The corresponding smoothness of graph signal $\mathbf{z}$ is
\begin{equation}
    \begin{aligned}
S_2(\mathbf{z})&=\frac{\mathbf{z}^T{\mathbf{L}}\mathbf{z}}{\mathbf{z}^T\mathbf{z}}=\frac{\mathbf{x}^T\mathbf{U}\mathbf{F}\mathbf{\Lambda}\mathbf{F}\mathbf{U}^T\mathbf{x}}{\mathbf{x}^T\mathbf{U}\mathbf{F}^2\mathbf{U}^T\mathbf{x}}=\frac{\sum_{i=1}^nf_i^2\lambda_iy_i^2}{\sum_{i=1}^nf_i^2y_i^2},
    \end{aligned}
\end{equation}
where $f_i=f(\lambda_i)$. Thus, the change of smoothness is
\begin{equation}
\begin{aligned}
    \Delta S&=S_2(\mathbf{z})-S_1(\mathbf{x})=\frac{\sum_{i=1}^n\sum_{j>i}^n(f_i^2-f_j^2)(\lambda_i-\lambda_j)y_i^2y_j^2}{\left(\sum_{i=1}^nf_i^2y_i^2\right)\left(\sum_{i=1}^ny_i^2\right)}.
\end{aligned}
\end{equation}
When $\mathcal{F}$ is a low-pass filter with $f_i > f_j$ for $\lambda_i<\lambda_j$, then
\begin{equation}
\begin{aligned}
    \Delta S&=\frac{\sum_{i=1}^n\sum_{j>i}^n(f_i^2-f_j^2)(\lambda_i-\lambda_j)y_i^2y_j^2}{\left(\sum_{i=1}^mf_i^2y_i^2\right)\left(\sum_{i=1}^ny_i^2\right)}<0,
\end{aligned}
\end{equation}
$\Delta S$<0 means the graph signal becomes smoother. Therefore, low-pass filter can smooth the graph signal.
When $\mathcal{F}$ is a high-pass filter with $f_i < f_j$ for $\lambda_i<\lambda_j$, then
\begin{equation}
    \begin{aligned}
\Delta S&=\frac{\sum_{i=1}^n\sum_{j>i}^n(f_i^2-f_j^2)(\lambda_i-\lambda_j)y_i^2y_j^2}{\left(\sum_{i=1}^mf_i^2y_i^2\right)\left(\sum_{i=1}^ny_i^2\right)}> 0.
\end{aligned}
\end{equation}
 $\Delta S$>0 means the graph signal becomes coarser. Therefore, high-pass filter can coarsen the graph signal.
\end{proof}

\subsubsection{Graph Convolution}
Instead of spatial domain, GSP processes graph signal in spectral domain. Therefore, the graph signal processing of a given graph signal $\mathbf{x}$ can be regarded as the graph convolution over the signal $\mathbf{x}$:
\begin{equation}
\mathbf{y}=\mathcal{F}\mathbf{x} = \mathbf{U}diag(f(\lambda_1), f(\lambda_2), \cdots,f(\lambda_n))\mathbf{U}^T\mathbf{x}.
\end{equation}

\begin{figure*}[ht!]
    \centering
    \includegraphics[width=1.0\linewidth]{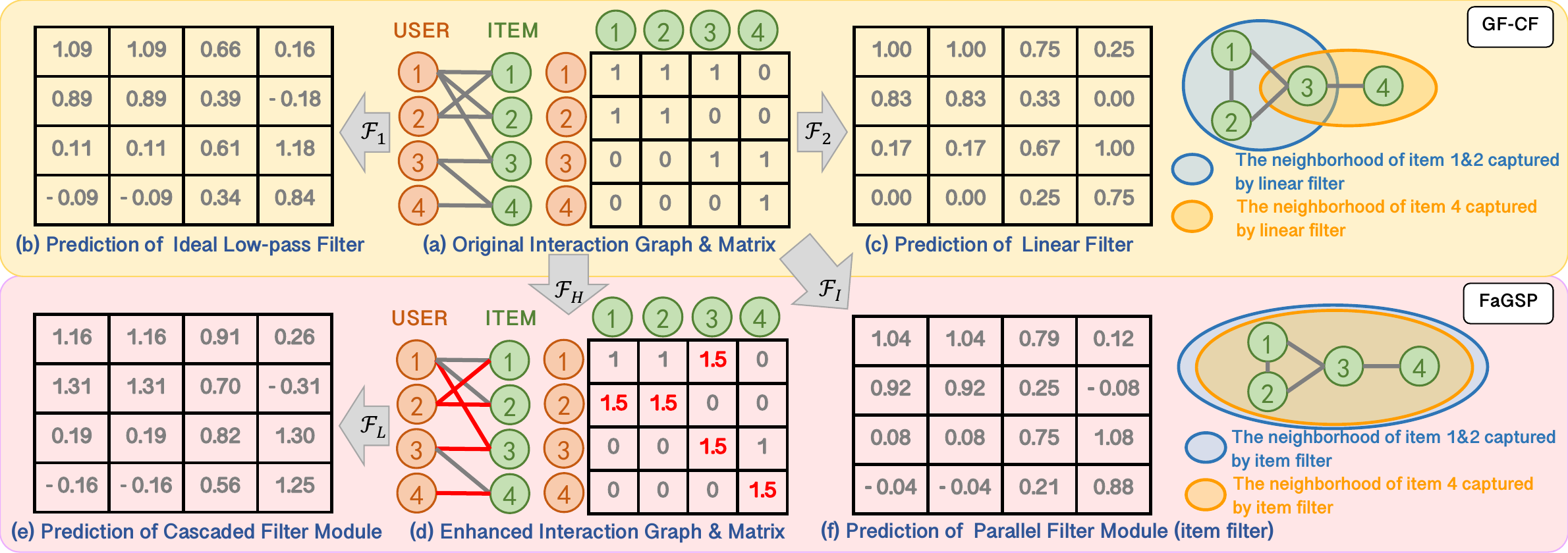}
    \caption{A toy example of user interaction prediction with 4 users and 4 items. (a) is the original interaction matrix, (b) and (c) are the prediction of the ideal low-pass filter and linear filter in GF-CF respectively. (d) is the enhanced interaction matrix, (e)--(f) are the predictions of the Cascaded Filter Module and Parallel Filter Module in \ours respectively. For ideal high-pass filter to enhance interactions in (d), we set $p_1=2$, $q=0.65$, and $\alpha_1=0.5$. For ideal low-pass filter in (b) and (d), we set $p_2=2$. In Parallel Filter Module, we only consider item high-order neighborhood information for the ease of presentation.}
    \label{fig:motivation1}
\end{figure*}

The signal $\mathbf{x}$ is first transformed from spatial domain to the spectral domain through Graph Fourier Transform basis $\mathbf{U}^T$, then the undesired frequencies are removed in the signal through function $f(\cdot)$ in spectral domain, and finally signal is transformed back to spatial domain through inverse Graph Fourier Transform basis $\mathbf{U}$.

\subsection{Notations}
Let the user set and item set be $\mathcal{U}$ and $\mathcal{V}$, and $|\mathcal{U}|=m$ and $|\mathcal{V}|=n$. The interactions between users and items can be represented as an interaction matrix $\mathbf{R}\in\{0,1\}^{m\times n}$, and if there is an interaction between user $u$ and item $i$, then $\mathbf{R}_{ui}=1$, otherwise, $\mathbf{R}_{ui}=0$. The normalized interaction matrix can be defined as $\tilde{\mathbf{R}}=\mathbf{D}_U^{-\frac{1}{2}}\mathbf{R}\mathbf{D}_I^{-\frac{1}{2}}$, where $\mathbf{D}_U=diag(\mathbf{R}\cdot\mathbf{1})$ and $\mathbf{D}_I=diag(\mathbf{R}^T\cdot\mathbf{1})$ are the user degree matrix and item degree matrix respectively. Then we can define user/item co-occurrence relationship matrix as follows:
\begin{equation}
\mathbf{O}_U=\tilde{\mathbf{R}}\tilde{\mathbf{R}}^T,\quad\mathbf{O}_I=\tilde{\mathbf{R}}^T\tilde{\mathbf{R}}.
\end{equation}

\subsection{The GF-CF and PGSP methods}
GF-CF~\cite{gfcf} adopted a combined filter to model user preference
\begin{equation}
\mathcal{F}=\alpha\mathbf{D}_I^{-\frac{1}{2}}\mathbf{V}_k\mathbf{V}_k^T\mathbf{D}_I^{\frac{1}{2}}+\mathbf{O}_I,\label{eq:gfcfeq}
\end{equation}
where $\mathbf{V}_k$ is the eigenvector matrix corresponding to the $k$ largest eigenvalues of $\mathbf{O}_I$, but in practice, it is obtained by perform Singular Value Decomposition (SVD)~\cite{svd} on $\tilde{\mathbf{R}}$ for high efficiency. $\alpha$ is a hyper-parameter that balances the two terms. Formally, GF-CF is composed of two types of filters: (1) an ideal low-pass filter $\mathcal{F}_1=\mathbf{D}_I^{-\frac{1}{2}}\mathbf{V}_k\mathbf{V}_k^T\mathbf{D}_I^{\frac{1}{2}}$, and (2) a linear filter $\mathcal{F}_2=\mathbf{O}_I$. PGSP~\cite{pgsp} designed a mixed-frequency filter to predict user future interactions
\begin{equation}
\mathcal{F}= (1-\phi)\bar{\mathbf{V}}_k\bar{\mathbf{V}}_k^T+\phi\mathbf{A},\label{eq:pgspeq}
\end{equation}
where $\bar{\mathbf{V}}_k$ is the eigenvector matrix corresponding to the first $k$ smallest eigenvalues of graph laplacian matrix $\mathbf{L}=\mathbf{I}-\mathbf{A}$, and $\mathbf{A}=\begin{bmatrix}\mathbf{O}_U&\tilde{\mathbf{R}}\\\tilde{\mathbf{R}}^T&\mathbf{O}_I\end{bmatrix}$ is the augmented similarity graph. PGSP is also composed of two filters: (1) an ideal low-pass filter $\mathcal{F}_1=\bar{\mathbf{V}}_k\bar{\mathbf{V}}_k^T$ and (2) a linear filter $\mathcal{F}_2=\mathbf{A}$.

\subsection{Motivation}\label{sec:motivation}
In this section, we take GF-CF method~\cite{gfcf} as an example to briefly illustrate the problems of existing GSP-based recommendation algorithms, which use an ideal low-pass filter and a linear filter to predict user future interactions, on a toy example composed of 4 users and 4 items. Figure~\ref{fig:motivation1} (b) and (c) show the results of user interaction prediction when using an ideal low-pass filter and a linear filter of GF-CF. We analyze the predictions as follows:

{\bf Ideal low-pass filter.} When predicting the interaction scores of $u_2$ (user 2) to $i_2$ and $i_3$ (item 2 and item 3), the difference between two scores is roughly the same as that between the interaction scores of $u_1$ to $i_2$ and $i_3$. This is because there is a connection between $i_2$ and $i_3$, that is, they are both interacted with by $u_1$, and the ideal low-pass filter can capture item common characteristic to smooth the graph signal, reducing the differences in predictions of $i_2$ and $i_3$. However, it ignores the item unique characteristic, that is, the target audience of $i_2$ is $u_1$ and $u_2$, while the target audience of $i_3$ is $u_1$ and $u_3$. Therefore, the difference in $u_2$'s interaction scores to $i_2$ and $i_3$ should be greater than that in $u_1$'s interaction scores to $i_2$ and $i_3$, and similar conclusions can be drawn for $u_3$. In addition, the relationship between $i_3$ and $i_4$ makes $u_4$ become a potential target audience of $i_3$, but not necessarily a potential target audience of $i_2$. Therefore, when predicting whether $u_4$ will interact with $i_2$ and $i_3$, the latter score should be significantly higher than the former.

{\bf Linear filter.} When predicting $u_4$'s future interactions, linear filters only focus on the interaction between $u_4$ and $i_3$, while ignoring that between $u_4$ and $i_1$(or $i_2$), although $i_1$ and $i_2$ also have a latent connection with the interacted  $i_4$, for example, $i_1$ and $i_4$ have an indirect relationship through $i_3$. Similarly, when predicting $u_2$'s future interactions, the filter only focuses on $i_3$, while ignoring $i_4$. This is because linear filters are constructed on the item co-occurrence matrix $\mathbf{O}_I$, and essentially can merely capture direct relationships between items. When using linear filters to predict user interactions, only items that have direct correlations to the user's interacted items can be predicted, while items that have indirect correlations cannot be predicted, resulting in limited coverage and inaccurate prediction results.  

Therefore, the issues of ideal low-pass filter and linear filter motivates us to design more suitable filters to extract richer information from user interactions, so as to make user preference modeling and interaction prediction accurately.

\section{The Proposed Method}\label{sec:method}
In this section, we introduce \ours, a frequency-aware graph signal processing method for collaborative filtering, to model user preference through two carefully designed filter modules: (1) a {\bf Cascaded Filter Module} used to model user preference using both unique and common characteristics, and  (2) a {\bf Parallel Filter Module} used to model user preference using user and item high-order neighborhood information. By combining these two modules, we can make the modeling of user preference more precise, thereby restore users' real interactive intentions and predict their future interactions more accurately.

\subsection{Cascaded Filter Module}\label{sec:cfm}
The proposition~\ref{prop:prop1} shows that low-pass filter can make the graph signal smoother, which is equal to retain the common characteristics among nodes and ignore the unique characteristics. Therefore, the ideal low-pass filters that existing GSP-based CF methods adopt can only capture common characteristics but fail to capture unique characteristics. Fortunately, ideal high-pass filter can capture unique characteristics since it coarsens graph signal to make nodes able to retain their own characteristics. Thus, we propose the Cascaded Filter Module, which is composed of an ideal high-pass filter and an ideal low-pass filter that work in a cascaded manner, to address the issue of the ideal low-pass filter. 

First, we perform SVD on the normalized interaction matrix $\tilde{\mathbf{R}}$
\begin{equation}
\tilde{\mathbf{R}}=\mathbf{U}\mathbf{\Lambda}\mathbf{V},\label{eq:svd}
\end{equation}
where singular values in $\mathbf{\Lambda}$ are arranged in descending order. 
We take the last $p_1$ rows of $\mathbf{V}$ (i,e., $\mathbf{V}_{-p_1:}$), which corresponds to the $p_1$ high frequency components, to construct the ideal high-pass filter as follows:
\begin{equation}
\mathcal{F}_H=\mathbf{D}^{-\frac{1}{2}}_I\mathbf{V}_{-p_1:}\mathbf{V}_{-p_1:}^T\mathbf{D}^{\frac{1}{2}}_I.\label{eq:hpf}
\end{equation}
Then, we exert $\mathcal{F}_H$ on the interaction signal $\mathbf{R}$ and obtain predicted matrix $\mathbf{R}^*$ whose element $\mathbf{R}^*_{u,i}$ is proportional to the probability that the corresponding interaction can reflect unique characteristics
\begin{equation}
    \mathbf{R}^* = \mathbf{R}\mathcal{F}_H.
\end{equation}
In order to filter out those interactions that truly reflect  unique characteristics in $\mathbf{R}^*$, we first calculate the $q$ quantile of $\mathbf{R}^*_{:, i}$ (the $i$-th column of $\mathbf{R}^*$) for each item $i$, which we denote as $r^q_i$. Then by comparing $\mathbf{R}^*_{:, i}$ and $r^q_i$, we can obtain those interactions that reflect unique characteristics, formed as $\mathbf{R}_H$ 
\begin{equation}
\begin{aligned}
    (\mathbf{R}_H)_{u,i} = \left\{\begin{matrix}
    1,&\text{~if~} \mathbf{R}^*_{u,i} \ge r^q_i~\text{and}~\mathbf{R}_{u,i}>0, \\
    0,&\text{~otherwise.}
    \end{matrix}\right.
\end{aligned}
\end{equation}
Note that these interactions should also be user historical interactions, i.e., $\mathbf{R}_{u,i}>0$. Finally, we can obtain the enhanced interaction signal $\hat{\mathbf{R}}$ by highlighting those interactions in $\mathbf{R}$ with $\mathbf{R}_H$
\begin{equation}
\hat{\mathbf{R}} = \mathbf{R} + \alpha_1\cdot\mathbf{R}_H,\label{eq:enhancement}
\end{equation}
where $\alpha_1\in\mathbb{R}^+$ is to control the impacts of  common characteristics and unique characteristics on user preference modeling and future interaction prediction. Smaller $\alpha_1$ will make the model focus on common characteristics, while larger $\alpha_1$ will emphasize the importance of unique characteristics. 

After obtaining the enhanced interaction signal $\hat{\mathbf{R}}$, we perform SVD on it to construct an ideal low-pass filter as follows:
\begin{equation}    \mathcal{F}_L=\hat{{\mathbf{D}}}^{-\frac{1}{2}}_I\hat{\mathbf{V}}_{:p_2}{\hat{\mathbf{V}}}_{:p_2}^T\hat{{\mathbf{D}}}^{\frac{1}{2}}_I,\label{eq:lpf}
\end{equation}
where $\hat{\mathbf{V}}_{:p_2}$ is constructed by the first $p_2$ rows of singular vector matrix $\hat{\mathbf{V}}$ of $\hat{\mathbf{R}}$, and $\hat{\mathbf{D}}_I=diag(\hat{\mathbf{R}}^T\cdot\mathbf{1})$. Then we use the filter to predict user future interactions as follows:
\begin{equation}
    \mathbf{P}_1 = \hat{\mathbf{R}}\mathcal{F}_L.
\end{equation}

It is worth noting that the ideal low-pass filter is designed on the enhanced interaction signal instead of the original interaction signal, thus the unique characteristics are taken into consideration when predicting user future interaction. Compared to merely using ideal low-pass filter ($\alpha_1=0$), Cascaded Filter Module can fully utilize unique characteristics and common characteristics from user interactions to jointly model user preference by exerting ideal high-pass filter and ideal low-pass filter in a cascaded manner, leading to more accurate interaction prediction.

The interactions marked in red in Figure~\ref{fig:motivation1} (d)---which are recognized by ideal high-pass filter---are those interactions that reflect unique characteristics. The results are reasonable, for example, the unique characteristic of $i_4$ should be reflected by its interaction with $u_4$ because there are no other users except $u_4$ interacting with $i_4$, which indicates the unique characteristics of $i_4$, e.g., the target audience such as $u_4$. Similarly, the interactions between $u_3$ and $i_3$ can reflect the unique characteristic of $i_3$.  The Figure~\ref{fig:motivation1} (e) shows the interaction prediction of Cascaded Filter Module in \ours. Comparing to Figure~\ref{fig:motivation1} (b), we can find that the difference in predicted scores of $u_4$ to $i_2$ and $i_3$ increases, and the score of the latter is significantly higher than that of the former, which aligns with our analysis in Section~\ref{sec:motivation}, that is, $i_3$ has a higher degree of matching with $u_4$'s preference compared to $i_2$.

\subsection{Parallel Filter Module}\label{sec:pfm}

\begin{figure}[t!]
    \centering
    \includegraphics[width=0.7\linewidth]{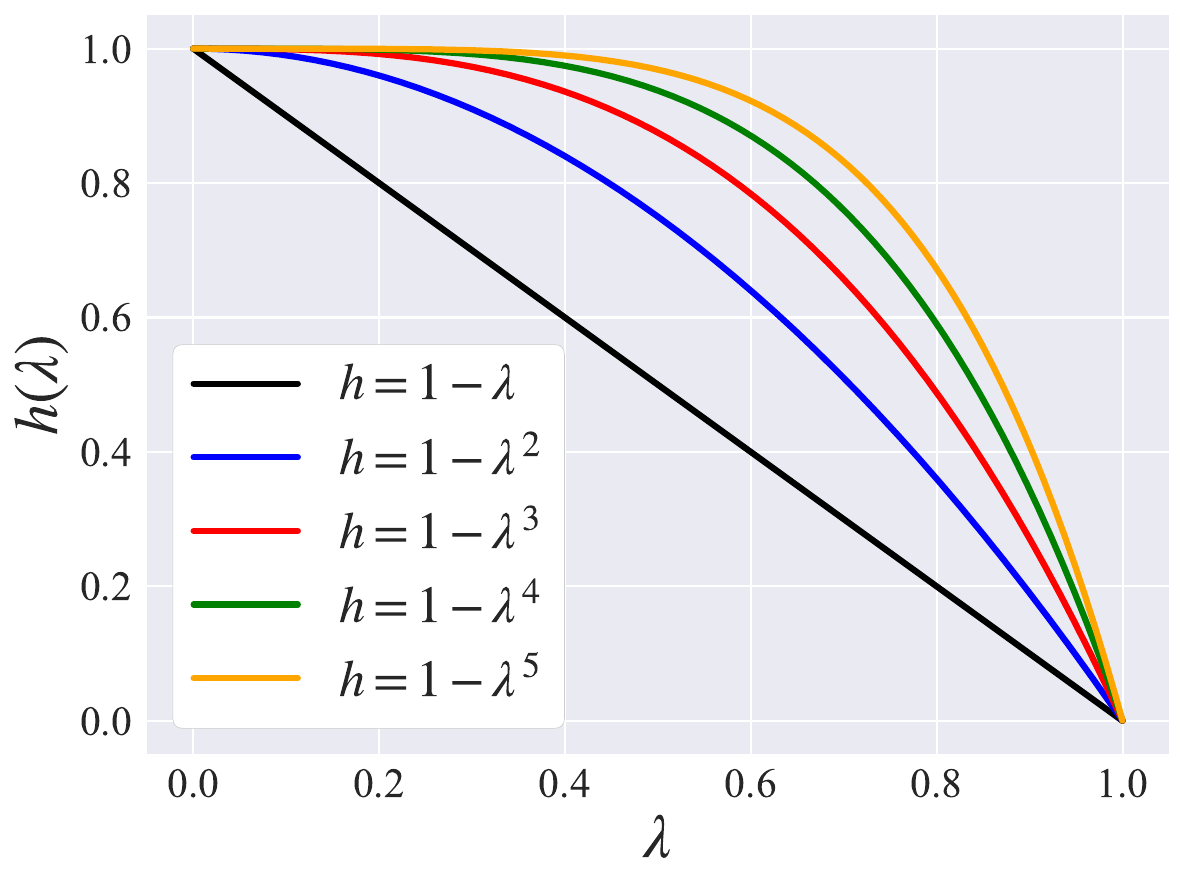}
    \caption{The frequency response functions $h(\lambda)=1-\lambda^{k_1}$ with respect to different order $k_1$ ranging from 1 to 5.}
    \label{fig:freq_func}
\end{figure}

To address the issue that existing GSP-based CF methods cannot capture user and item high-order neighborhood information, we propose the Parallel Filter Module to capture user and item high-order neighborhood characteristics respectively. Since the process of extracting user and item higher-order neighborhood information is the same, we take item as an example to illustrate how to extract this characteristic. We design the item low-pass filter as follows:
\begin{equation}
    \mathcal{F}_{I}=\mathbf{I} - (\mathbf{I}-\mathbf{O}_I)^{k_1},\label{eq:irf}
\end{equation}
where $\mathbf{O}_I=\tilde{\mathbf{R}}^T\tilde{\mathbf{R}}$, and $k_1 (k_1\ge 2)$ is the order of the low-pass filter and controls the range of neighborhood information extraction. The design of the filter is based on the following two considerations: one is that it can capture neighborhood information with any order. As $k_1$ increases, more distant neighborhood information can be extracted. Note that when $k_1=1$, $\mathcal{F}_{I}$ degenerates to a linear filter $\mathcal{F}=\mathbf{O}_I$ that existing GSP-based CF methods adopt. 

Proposition~\ref{prop:prop2} shows that the frequency response function of $\mathcal{F}_{I}$ is a non-linear concave function, which reveals the other consideration for the design of the filter, i.e., preserving the common and unique characteristics in the high-order neighborhood of item.

\begin{proposition}\label{prop:prop2}
The frequency response function of $\mathcal{F}_{I}$ is a non-linear concave function.
\end{proposition}
\begin{proof}
    Let $\mathbf{L}_I=\mathbf{I}-\mathbf{O}_I$. Since both $\mathbf{I}$ and $\mathbf{O}_I$ are real and symmetric matrices, $\mathbf{L}_I$ is also a real and symmetric matrix, thus $\mathbf{L}_I$ is diagonalizable and can be represented as $\mathbf{L}_I=\mathbf{U}\pmb\Lambda\mathbf{U}^T$, and $\mathbf{U}\mathbf{U}^T=\mathbf{U}^T\mathbf{U}=\mathbf{I}$. Then we have
    \begin{align}
        \mathcal{F}_{I} &= \mathbf{I} - (\mathbf{I}-\mathbf{O}_I)^{k_1}=\mathbf{U}\mathbf{U}^T - \mathbf{U}\pmb\Lambda^{k_1}\mathbf{U}^T=\mathbf{U}(\mathbf{I}-\pmb\Lambda^{k_1})\mathbf{U}^T\label{eq:itemnonlinear}.
    \end{align}
    By comparing Eq. (\ref{eq:filter}) and Eq. (\ref{eq:itemnonlinear}), we can easily find that the frequency response function of $\mathcal{F}_{I}$ is $f(\lambda_i)=1-\lambda_i^{k_1}$, where $\lambda_i$ is the $i$-th eigenvalue of $\mathbf{L}_{I}$, and its range is [0,1], which is proved in GF-CF~\cite{gfcf}. When $k_1\ge2$, it is obviously a non-linear concave function, since for any $\lambda_i, \lambda_j\in[0,1]$ and $\beta\in[0,1]$, we have $f\left((1-\beta)\times\lambda_i+\beta\times\lambda_j\right)\ge(1-\beta)\times f(\lambda_i)+\beta\times f(\lambda_j)$ .
\end{proof}

Figure~\ref{fig:freq_func} shows the frequency response functions with respect to different order of filter $k_1$ ranging from 1 to 5. We can find that when $k_1=1$, the frequency response function of linear filter (black curve) is linear, both low frequency and high frequency are attenuated, making it unable to capture common and unique characteristics in direct neighborhood of item for interaction prediction. However, when $k_1$ grows, the frequency response function (colored curves) becomes non-linear, and more and more low frequency components are preserved and high frequency components are enhanced. For example, when $k_1=2$, around 3.2\% low frequency components can be retained since $h(0.032)=1-0.032^2=0.999\approx1$, while when $k_1=5$, around 25.1\% low frequency components can be retained. Preserving more low frequency components and high frequency components means more information with respect to items in high-order neighborhood can be used to model user preference, thus achieving more accurate interaction prediction.

With item low-pass filter $\mathcal{F}_{I}$, we can predict user future interactions based on the item high-order neighborhood information:
\begin{equation}
\mathbf{P}_2=\mathbf{R}\mathcal{F}_{I}.
\end{equation}

We can also predict user future interactions based on user high-order neighborhood information in the user interactions by exerting a user low-pass filter as follows:
\begin{align}
\mathbf{O}_U=\tilde{\mathbf{R}}\tilde{\mathbf{R}}^T,\quad
\mathcal{F}_{U}=\mathbf{I}-(\mathbf{I}-\mathbf{O}_U)^{k_2},\quad\label{eq:urf}
\mathbf{P}_3=\mathcal{F}_{U}\mathbf{R},
\end{align}
where $k_2 (k_2\ge2)$ is the order of filter $\mathcal{F}_{U}$.

The Figure~\ref{fig:motivation1} (f) shows the interaction prediction of Parallel Filter Module in \ours. For the ease of presentation, we only use item low-pass filter $\mathcal{F}_I$ to predict interactions. Comparing to Figure~\ref{fig:motivation1} (c), we can find that the neighborhood of $i_1$ has expanded from $i_2$ and $i_3$ to $i_2$, $i_3$ and $i_4$, and the the neighborhood of $i_4$ has expanded from $i_3$ to $i_1$, $i_2$ and $i_3$, which achieves the effect of taking more items into consideration for user interaction prediction, for example,  $i_4$ for $u_2$, and $i_1$ and $i_2$ for $u_4$. In addition, due to the difference between $i_4$ and $i_1$ (or $i_2$) and the relationship between $i_3$ and $i_1$ (or $i_2$) , after introducing high-order neighborhood information, the prediction score of $u_4$ to $i_3$ decreases. This indicates that high-order neighborhood information can also correct the prediction deviation caused by merely using direct neighborhood information, making the interaction prediction more accurate.

\subsection{Model Inference}\label{sec:mi}
User/item common characteristics and unique characteristics, along with user and item high-order neighborhood information, are orthogonal but both beneficial for predicting user interactions.  
Therefore, we use a linear model to combine the outputs of Cascaded Filter Module and Parallel Filter Module to predict user interactions 
\begin{equation}
    \mathbf{P} = \alpha_2\cdot\mathbf{P_1}+\mathbf{P_2}+\mathbf{P_3},
\end{equation}
where $\alpha_2$ is a coefficient that balance different prediction terms.

\subsection{The Time Complexity}

In this section, we take FLoating-point OPerations (FLOPs) as a quantification metric to analyze the time complexity of \ours. As shown in~\citet{flops}, the FLOPs of truncated SVD on $\mathbf{A}\in\mathbb{R}^{r\times s}$ is $\mathcal{O}(2rs^2+s^3)$. Therefore, the FLOPs for $\mathcal{F_H}$ and $\mathcal{F}_L$ are $\mathcal{O}(2mn^2+n^3)+\mathcal{O}(p_1n^2)$ and $\mathcal{O}(2mn^2+n^3)+\mathcal{O}(p_2n^2)$, and that for $\mathbf{R}^*$ and $\mathbf{P}_1$ are both $\mathcal{O}(mn^2)$. Thus the total FLOPs of Cascaded Filter Module is $\mathcal{O}(mn^2+n^3)$ where $p_1, p_2\ll n$. The FLOPs for $\mathbf{O}_I$, $\mathcal{F}_I$ and $\mathbf{P}_2$ are $\mathcal{O}(mn^2)$,  $\mathcal{O}(n^2+(k_1-1)\times n^3+n^2)\approx\mathcal{O}(n^3)$ and $\mathcal{O}(mn^2)$, similarly, that for $\mathcal{O}_U$,  $\mathcal{F}_U$ and $\mathbf{P}_3$ are $\mathcal{O}(m^2n)$, $\mathcal{O}(m^3)$ and $\mathcal{O}(m^2n)$. Thus the total FLOPs of Parallel Filter Module is $\mathcal{O}(mn^2+n^3+m^3+m^2n)$. To sum up, if $m\approx n$, then the total FLOPs of \ours is $\mathcal{O}(n^3)$, which is similar to that of GF-CF~\cite{gfcf}.

\section{Experiment}
\subsection{Experimental Setting}

\begin{table}[t!]
\centering
\caption{The statistics of the six real-world datasets.}
\resizebox{1.0\linewidth}{!}{
\begin{tabular}{l|cccc|c}
\toprule
& \# Users & \# Items & \# Interactions & Density & Domain\\
\midrule
ML100K (MovieLens 100K)&943&1,682&100,000&0.0630 & Movie\\
Beauty (Amazon Beauty)&22,363&12,101&198,502&0.0007&Product\\
BX (Book-Crossing)&18,964&19,998&482,153&0.0013&Book\\
LastFM&992&10,000&571,817&0.0576&Music\\
ML1M (MovieLens 1M)&6,040&3,706&1,000,209&0.0477&Movie\\
Netflix&20,000&17,720&5,678,654&0.0160&Movie\\
\bottomrule
\end{tabular}
}
\label{tab:dataset}
\end{table}

We conduct experiments on six widely used datasets with varying domains: (1) {\bf ML100K, ML1M and Netflix}.(2) {\bf Beauty}. (3) {\bf BX}. (4) {\bf LastFM}. All datasets are divided into training set, validation set and test set with the ratio of 72\%:8\%:20\%. Table~\ref{tab:dataset} shows the statistics.

We evaluate the performance of \ours with three metrics in the Top-K scenario: (1) {\bf F1}, 
(2) {\bf Mean Reciprocal Rank (MRR)},
and (3) {\bf Normalized Discounted Cumulative Gain (NDCG)}. For each metric, we report their results when K=10 and K=20. 

We compare the performance of \ours with seven GCN-based methods and two GSP-based methods: (1) {\bf LR-GCCF} ~\cite{lrgccf}; (2) {\bf LCFN} ~\cite{lcfn}; (3) {\bf DGCF}~\cite{dgcf}; (4) {\bf LightGCN}~\cite{lightgcn}; (5) {\bf IMP-GCN}~\cite{imp-gcn}; (6) {\bf SimpleX}~\cite{simplex}; (7) {\bf UltraGCN}~\cite{ultragcn}; (8) {\bf GF-CF}~\cite{gfcf}; and (9) {\bf PGSP}~\cite{pgsp}. 

\begin{table*}[t!]
\setlength{\tabcolsep}{4.56pt}
\caption{The results of performance comparison on six public datasets. The best performance is denoted in bold, the second best performance is denoted with an underline. RI refers to relative improvement of \ours with respect to the best baseline. Note that DGCF occurred the Out-Of-Memory (OOM) problem on Netflix dataset, so we do not report the results.}
\label{tab:acc_com}
\centering
\resizebox{0.9\linewidth}{!}
{
\begin{tabular}{l|l|ccccccccc|c|c}
\toprule
 &&LR-GCCF&LCFN&DGCF&LightGCN&IMP-GCN&SimpleX&UltraGCN&GF-CF&PGSP&\ours & RI\\
\midrule
\multirow{6}{*}{ML100K}&F1@10&$0.1444$&$0.1393$&$0.2421$&$0.2461$&$0.2287$&$\underline{0.2466}$&$0.2366$&$0.2425$&$0.2407$&$\mathbf{0.2600}$&+5.43\%\\
&MRR@10&$0.4616$&$0.4142$&$0.6012$&$0.5877$&$0.5690$&$\underline{0.6064}$&$0.5749$&$0.6010$&$0.5767$&$\mathbf{0.6380}$&+5.21\%\\
&NDCG@10&$0.5603$&$0.5326$&$0.6808$&$0.6771$&$0.6605$&$\underline{0.6866}$&$0.6688$&$0.6843$&$0.6722$&$\mathbf{0.7092}$&+3.29\%\\
&F1@20&$0.1974$&$0.1759$&$0.3157$&$\underline{0.3248}$&$0.1790$&$0.3163$&$0.3145$&$0.3154$&$0.3132$&$\mathbf{0.3348}$&+3.08\%\\
&MRR@20&$0.4411$&$0.3557$&$0.5422$&$0.5652$&$\underline{0.5690}$&$0.5411$&$0.5289$&$0.5593$&$0.5496$&$\mathbf{0.5901}$&+3.71\%\\
&NDCG@20&$0.5555$&$0.5043$&$0.6538$&$0.6618$&$0.6605$&$0.6543$&$0.6486$&$\underline{0.6621}$&$0.6495$&$\mathbf{0.6805}$&+2.78\%\\ 
\midrule
\multirow{6}{*}{Beauty}&F1@10&$0.0268$&$0.0094$&$0.0325$&$0.0327$&$0.0289$&$0.0333$&$0.0287$&$0.0330$&$\underline{0.0335}$&$\mathbf{0.0351}$&+4.78\%\\
&MRR@10&$0.0443$&$0.0147$&$0.0507$&$0.0519$&$0.0451$&$0.0517$&$0.0441$&$0.0517$&$\underline{0.0525}$&$\mathbf{0.0549}$&+4.57\%\\
&NDCG@10&$0.0635$&$0.0231$&$0.0750$&$0.0763$&$0.0667$&$0.0761$&$0.0649$&$0.0751$&$\underline{0.0765}$&$\mathbf{0.0799}$&+4.44\%\\
&F1@20&$0.0217$&$0.0078$&$0.0262$&$0.0265$&$0.0238$&$0.0271$&$0.0229$&$\underline{0.0273}$&$0.0272$&$\mathbf{0.0284}$&+4.03\%\\
&MRR@20&$0.0388$&$0.0165$&$0.0468$&$0.0466$&$0.0409$&$0.0461$&$0.0393$&$0.0469$&$\underline{0.0481}$&$\mathbf{0.0485}$&+0.83\%\\
&NDCG@20&$0.0678$&$0.0295$&$0.0823$&$\underline{0.0829}$&$0.0729$&$0.0819$&$0.0689$&$0.0819$&$0.0828$&$\mathbf{0.0846}$&+2.05\%\\ 
\midrule
\multirow{6}{*}{BX}
&F1@10&$0.0142$&$0.0140$&$0.0319$&$0.0334$&$0.0155$&$\underline{0.0346}$&$0.0319$&$0.0309$&$0.0308$&$\mathbf{0.0361}$&+4.34\%\\
&MRR@10&$0.0274$&$0.0276$&$0.0580$&$0.0601$&$0.0331$&$\underline{0.0628}$&$0.0594$&$0.0576$&$0.0578$&$\mathbf{0.0656}$&+4.46\%\\
&NDCG@10&$0.0414$&$0.0411$&$0.0839$&$0.0861$&$0.0479$&$\underline{0.0892}$&$0.0847$&$0.0821$&$0.0817$&$\mathbf{0.0931}$&+4.37\%\\
&F1@20&$0.0138$&$0.0138$&$0.0299$&$0.0299$&$0.0157$&$\underline{0.0315}$&$0.0282$&$0.0293$&$0.0293$&$\mathbf{0.0328}$&+4.13\%\\
&MRR@20&$0.0286$&$0.0251$&$0.0563$&$0.0546$&$0.0305$&$\underline{0.0588}$&$0.0535$&$0.0543$&$0.0553$&$\mathbf{0.0607}$&+3.23\%\\
&NDCG@20&$0.0495$&$0.0468$&$0.0951$&$0.0926$&$0.0534$&$\underline{0.0976}$&$0.0913$&$0.0917$&$0.0923$&$\mathbf{0.1006}$&+3.07\%\\
\midrule
\multirow{6}{*}{LastFM}&F1@10&$0.0689$&$0.0507$&$0.0867$&$\underline{0.0968}$&$0.0799$&$0.0941$&$0.0811$&$0.0964$&$0.0945$&$\mathbf{0.0995}$&+2.79\%\\
&MRR@10&$0.5046$&$0.4280$&$0.5706$&$0.6073$&$0.5398$&$0.5930$&$0.5529$&$\underline{0.6086}$&$0.6007$&$\mathbf{0.6338}$&+4.36\%\\
&NDCG@10&$0.5893$&$0.5171$&$0.6509$&$\underline{0.6855}$&$0.6241$&$0.6690$&$0.6340$&$0.6817$&$0.6767$&$\mathbf{0.6969}$&+1.66\%\\
&F1@20&$0.1088$&$0.0796$&$0.1408$&$\underline{0.1532}$&$0.1243$&$0.1464$&$0.1298$&$0.1487$&$0.1472$&$\mathbf{0.1593}$&+3.98\%\\
&MRR@20&$0.4761$&$0.4067$&$0.5539$&$\underline{0.5840}$&$0.5175$&$0.5570$&$0.5332$&$0.5749$&$0.5653$&$\mathbf{0.5985}$&+2.48\%\\
&NDCG@20&$0.5819$&$0.5142$&$0.6467$&$\underline{0.6683}$&$0.6179$&$0.6519$&$0.6246$&$0.6639$&$0.6593$&$\mathbf{0.6836}$&+2.29\%\\  
\midrule
\multirow{6}{*}{ML1M}&F1@10&$0.0930$&$0.0860$&$0.1956$&$0.2044$&$0.1837$&$0.2087$&$0.1963$&$\underline{0.2106}$&$0.2090$&$\mathbf{0.2203}$&+4.61\%\\
&MRR@10&$0.2946$&$0.2839$&$0.4825$&$0.5010$&$0.4685$&$0.5051$&$0.4886$&$0.4996$&$\underline{0.5063}$&$\mathbf{0.5229}$&+3.28\%\\
&NDCG@10&$0.3787$&$0.3694$&$0.5740$&$0.5873$&$0.5594$&$0.5921$&$0.5773$&$0.5897$&$\underline{0.5923}$&$\mathbf{0.6082}$&+2.68\%\\
&F1@20&$0.1246$&$0.1199$&$0.2448$&$0.2534$&$0.2273$&$0.2546$&$0.2435$&$\underline{0.2570}$&$0.2548$&$\mathbf{0.2673}$&+4.01\%\\
&MRR@20&$0.2924$&$0.2529$&$0.4476$&$0.4644$&$0.4292$&$0.4625$&$0.4500$&$0.4557$&$\underline{0.4646}$&$\mathbf{0.4748}$&+2.20\%\\
&NDCG@20&$0.3980$&$0.3809$&$0.5586$&$\underline{0.5692}$&$0.5416$&$0.5687$&$0.5600$&$0.5678$&$0.5687$&$\mathbf{0.5808}$&+2.04\%\\  
\midrule
\multirow{6}{*}{Netflix}&F1@10&$0.0920$&$0.0660$&$--$&$0.1205$&$0.1026$&$0.1137$&$0.0835$&$\underline{0.1244}$&$0.1185$&$\mathbf{0.1294}$&+4.02\%\\
&MRR@10&$0.5234$&$0.4056$&$--$&$0.5974$&$0.5386$&$0.5839$&$0.4708$&$\underline{0.6126}$&$0.5905$&$\mathbf{0.6293}$&+2.73\%\\
&NDCG@10&$0.6134$&$0.5058$&$--$&$0.6814$&$0.6307$&$0.6694$&$0.5746$&$\underline{0.6928}$&$0.6756$&$\mathbf{0.7079}$&+2.18\%\\
&F1@20&$0.1432$&$0.1043$&$--$&$0.1848$&$0.1569$&$0.1784$&$0.1227$&$\underline{0.1890}$&$0.1792$&$\mathbf{0.1970}$&+4.23\%\\
&MRR@20&$0.4963$&$0.3852$&$--$&$0.5636$&$0.4995$&$0.5555$&$0.4309$&$\underline{0.5691}$&$0.5502$&$\mathbf{0.5827}$&+2.39\%\\
&NDCG@20&$0.6043$&$0.5162$&$--$&$0.6642$&$0.6152$&$0.6572$&$0.5580$&$\underline{0.6689}$&$0.6543$&$\mathbf{0.6814}$&+1.87\%\\  
\bottomrule
\end{tabular}
}
\end{table*}

For all baselines, we use their released code and carefully tune hyper-parameters according to their papers. For \ours, we tuned the number of high and low frequency components in ideal high-pass filter and ideal low-pass filter $p_1$ and $p_2$ from 16 to 256 for ML100K, and 32 to 1024 for other datasets. We tune $\alpha_1$ and $\alpha_2$ from 0.1 to 1.0 with step 0.05. For the orders of item low-pass filter and user low-pass filter $k_1$ and $k_2$, we tune them from 2 to 14, and for quantile $q$, we tune it from [0.6, 0.65, 0.7, 0.75, 0.8]. Note that in practice, we can tune each hyper-parameter in a hierarchical manner to reduce the search space of hyper-parameters, for instance, we can first search $k_1$ and $k_2$ in [2, 6, 10, 14], then in [11, 12, 13] on ML1M dataset.

\begin{table*}[t!]
\setlength{\tabcolsep}{2pt}
\caption{The ablation study of \ours on ML100K, ML1M and LastFM datasets. The best performance is denoted in bold. 
}
\centering
\resizebox{1.0\linewidth}{!}
{
\begin{tabular}{l|cccc|cccc|cccc}
 \toprule
 \multirow{2}{*}{}&\multicolumn{4}{c|}{\bf ML100K}&\multicolumn{4}{c|}{\bf ML1M}&\multicolumn{4}{c}{\bf LastFM}\\
  \cline{2-5}\cline{6-9}\cline{10-13}
&MRR@10&NDCG@10&MRR@20&NDCG@20&MRR@10&NDCG@10&MRR@20&NDCG@20&MRR@10&NDCG@10&MRR@20&NDCG@20\\
\midrule
(1)~\ours&$\mathbf{0.6380}$&$\mathbf{0.7092}$&$\mathbf{0.5901}$&$\mathbf{0.6805}$&$\mathbf{0.5229}$&$\mathbf{0.6082}$&$\mathbf{0.4748}$&$\mathbf{0.5808}$&$\mathbf{0.6338}$&$\mathbf{0.6969}$&$\mathbf{0.5985}$&$\mathbf{0.6836}$\\
\midrule
(2)~\ours w/o IHF&0.6233&0.6998&0.5877&0.6795&0.5196&0.6062&0.4744&$\mathbf{0.5808}$&0.6249&0.6937&0.5952&0.6810\\
(3)~\ours w/o IHF+ILF&0.6227&0.6993&0.5856&0.6791&0.5128&0.6025&0.4683&0.5775&0.6234&0.6929&0.5817&0.6723\\
\midrule
(4)~\ours w/o IHNF&0.6187&0.6976&0.5791&0.6733&0.5117&0.5988&0.4722&0.5784&0.6242&0.6946&0.5905&0.6757\\
(5)~\ours w/o UHNF&0.6313&0.7046&0.5836&0.6749&0.5163&0.6044&0.4726&0.5783&0.6304&0.6938&0.5892&0.6780\\
(6)~\ours w/o IHNF+UHNF&0.5938&0.6760&0.5694&0.6611&0.4918&0.5806&0.4543&0.5620&0.6063&0.6827&0.5651&0.6640\\
\bottomrule
\end{tabular}
}
\label{tab:ablstu}
\end{table*}

\subsection{Performance Comparison}
Table~\ref{tab:acc_com} compares the performance between \ours and other methods on six datasets. Specifically, 
we have the following observations:
\begin{enumerate}[0]
    \item[$\bullet$] LightGCN and SimpleX achieve the best performance among all GCN-based CF methods substantially. This is because LightGCN removes feature transformation and non-linear activation that will hurt the expressivity of GCN, thereby improving the accuracy of user preference modeling, and SimpleX designs a cosine contrastive loss with large negative sampling ratio to train the model, making it able to distill more information from supervision signal, thus achieving more accurate interaction prediction.
    
    \item[$\bullet$] GSP-based methods, i.e., GF-CF and PGSP, can achieve better performance than GCN-based methods in most cases. The reason is GCN-based methods use low-pass filters whose frequency response function is non-linear convex~\cite{gcnlpf}, and cannot preserve sufficient number of components to model accurate user preference. While GF-CF and PGSP use ideal low-pass filters, which can retain adaptive amount of low frequency components, to model user preference, making the preference more accurate.
    
    \item[$\bullet$] \ours obtain the best results on all datasets. This is because \ours can fully utilize not only user/item common characteristics and unique characteristics but also user and item high-order neighborhood information from user interactions, leading to the more precise user preference modeling and more accurate prediction of user future interactions.
\end{enumerate}

\subsection{Ablation Study}

We conduct ablation study on ML100K, ML1M and LastFM datasets to comprehensively analyze the effect of each component, i.e., Ideal High-pass Filter (IHF) in Eq. (\ref{eq:hpf}), Ideal Low-pass Filter (ILF) in Eq. (\ref{eq:lpf}), Item High-order Neighborhood Filter (IHNF) in Eq. (\ref{eq:irf}) and User High-order Neighborhood Filter (UHNF) in Eq. (\ref{eq:urf}), to the performance of \ours. Table~\ref{tab:ablstu} shows the experimental results. Due to the space limitation, we do not report the results of F1@10 and F1@20, while they have the same trends as other metrics. From the results, we have the following findings:
\begin{enumerate}[1.]
    \item Comparing setting (1) and (2), we can find that when removing IHF, the performance of \ours has decreased. This is because IHF can capture user/item unique characteristics, which can provide detailed information about which type of item that user will interact with. Therefore, neglecting unique characteristics, i.e., removing IHF, will hurt the performance of \ours.
    \item Comparing setting (2) and (3), the accuracy of \ours has decreased, which shows the effectiveness of ILF. The role of ILF is to capture user/item common characteristics, so as to provide general information on the direction of user future interaction. When removing ILF, the common characteristics cannot be utilized, thus affecting the recommendation performance.
    \item Comparing setting (1) and (4) or (5) and (6), the performance of \ours has declined, which demonstrates the effectiveness of IHNF to the modeling of user preference. IHNF can capture the high-order relationship of items, making \ours able to utilize information from item high-order neighbors so as to improve the accuracy of interaction prediction. UHNF can reach similar conclusions by comparing setting (1) and (5) or (4) and (6).
\end{enumerate}

\begin{figure*}[t!]
\begin{minipage}[t]{0.24\linewidth}
\centering
    \includegraphics[width=1.0\linewidth]{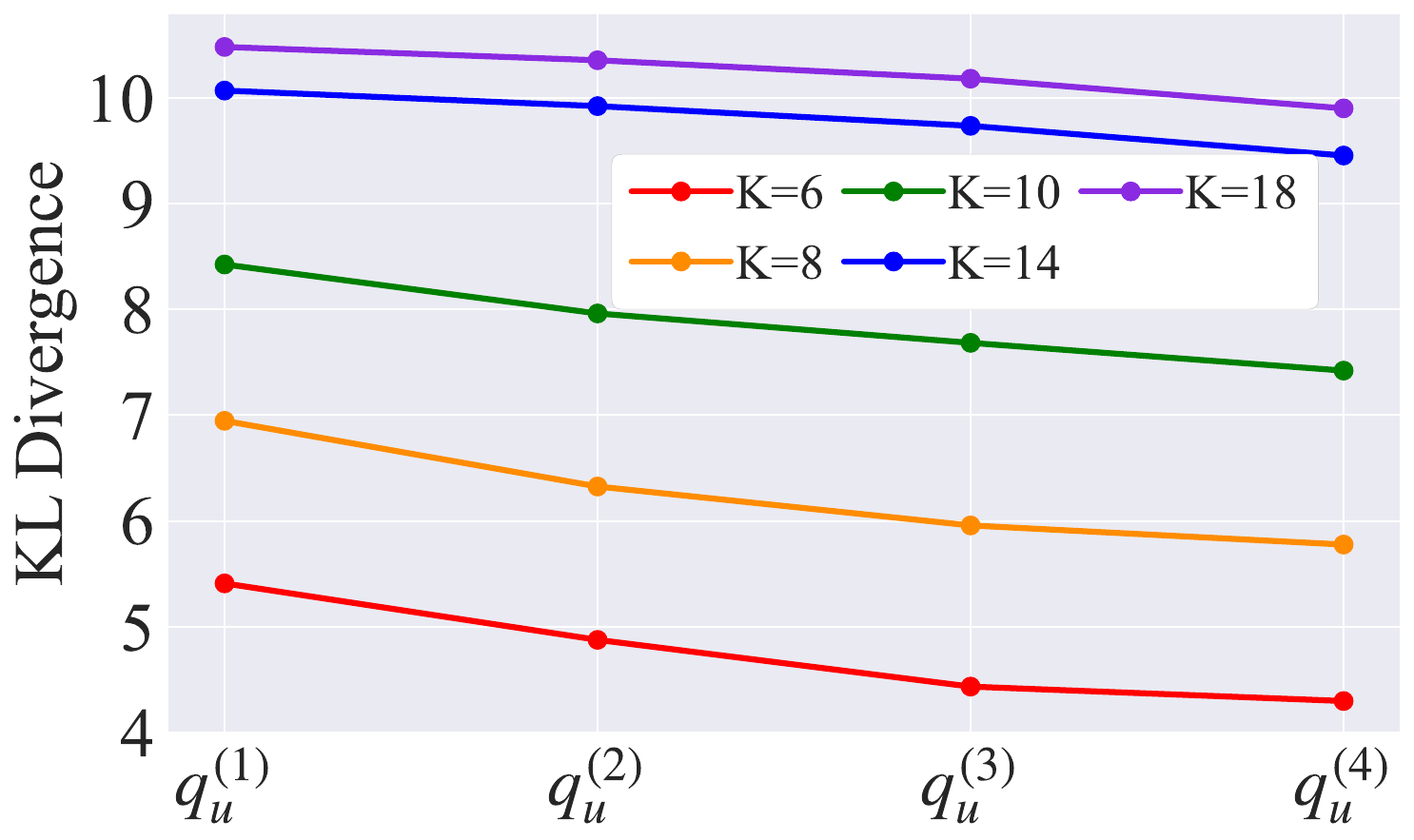}
    \caption{Visualization of consistency between user historical and predicted preference.
    }
    \label{fig:visualization}
\end{minipage}
\hfill
\begin{minipage}[t]{0.24\linewidth}
\centering
    \includegraphics[width=1.0\linewidth]{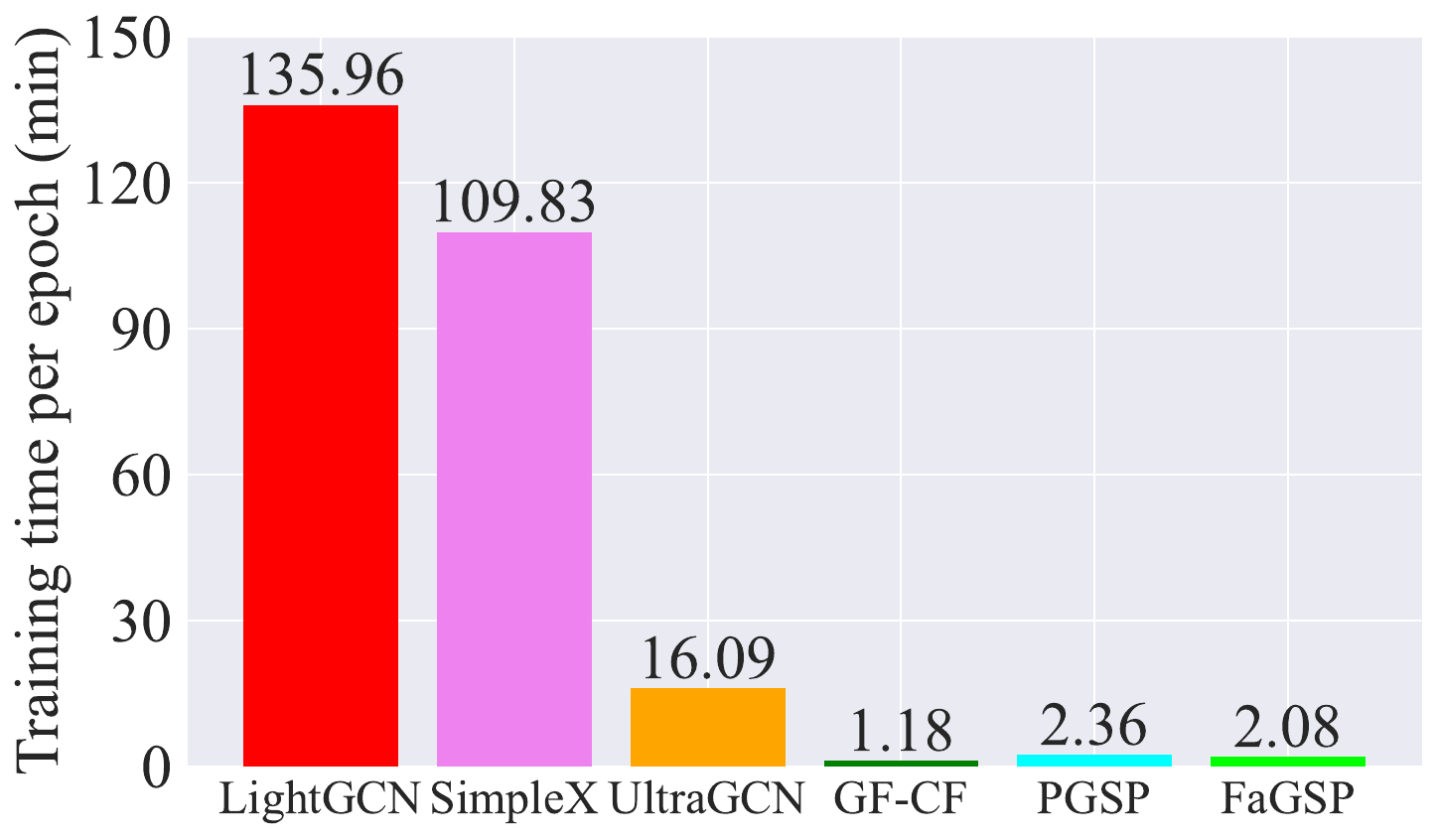}
    \caption{The average training time (5 times) of \ours and other methods on ML1M.}
    \label{fig:eff_ana}
\end{minipage}
\hfill
\begin{minipage}[t]{0.48\linewidth}
\centering
    \includegraphics[width=1.0\linewidth]{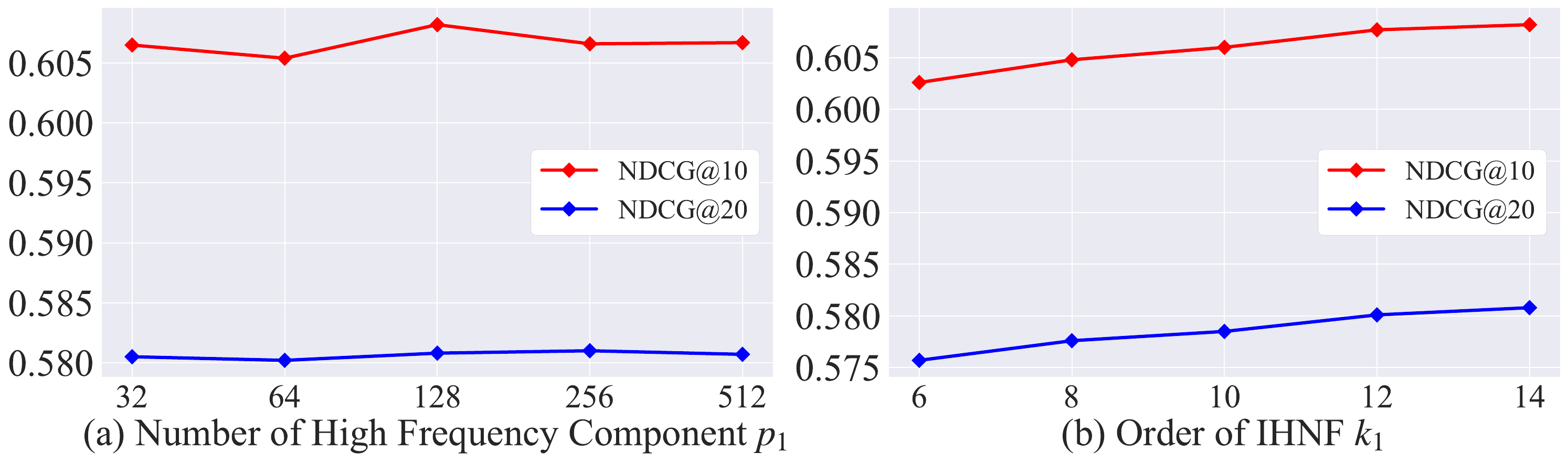}
    \caption{The sensitivity analysis of two hyper-parameters on ML1M dataset: Number of High Frequency Component $p_1$ and Order of Item High-order Neighborhood Filter (IHNF) $k_1$.}
    \label{fig:sen_ana}
\end{minipage}
\end{figure*}

\subsection{Visualization}
We visualize the consistency between user historical preference distribution (calculated from training data) and user predicted preference distribution (calculated from model prediction) on ML100K to show that introducing IHF, UHNF and IHNF can improve the accuracy of user preference modeling. Specifically, we define user $u$'s historical preference distribution $p_u$ according to the categories (e.g., Comedy, Action) of items that he/she has interacted with:
\begin{equation}
    p_{u}(\text{category}=l)=\frac{C_{ul}}{\sum_{k=1}^LC_{uk}},
\end{equation}
where $C_{ul}~(u=1,2,\cdots,M,~~~ l=1,2,\cdots, L)$ represents the number of appearances of the $l$-th category in the user $u$'s interacted items, $M$ is the number of users, and $L$ is the number of categories. 

Similarly, we can define user $u$'s predicted preference distributions from his/her predicted items in test using ILF as $q^{(1)}_u$, and that after introducing IHF, UHNF and IHNF sequentially on the basis of ILF as $q^{(2)}_u$ (i.e. ILF+IHF), $q^{(3)}_u$ (i.e. ILF+IHF+UHNF) and $q_u^{(4)}$ (i.e. ILF+IHF+UHNF+IHNF). Then we use the Kullback–Leibler (KL) divergence to evaluate the quality of the predicted preference distribution, where smaller KL divergence indicates better predicted preference distribution. The KL divergence between $p$ and $q$ is:
\begin{equation}
    \mathbf{KL}(p,q^{(w)}) = \frac{1}{M}\sum_{u=1}^M\sum_{l=1}^Lp_u(l)\ln\frac{p_u(l)}{q^{(w)}_u(l)},~w=1,2,3,4.
\end{equation}
If $\mathbf{KL}(p, q^{(1)}) > \mathbf{KL}(p, q^{(2)})$, it indicates that it is useful to introduce IHF, making the predicted preference distribution have a higher consistency with historical preference distribution. 

Figure~\ref{fig:visualization} shows the results of consistency between user historical preference and user predicted preference with respect to top-$K$ categories ($K\in\{6,8,10,14,18\}$) under different filter combinations. We can find that the KL divergence is constantly decreasing after introducing IHF, UHNF and IHNF sequentially on the basis of the ILF in all settings of $K$, which shows that IHF, UHNF and IHNF are all beneficial to restore user's real interactive intentions by utilizing user/item unique characteristics and user and item high-order neighborhood information, making the predicted user preference as close as possible to user historical preference, thereby achieving accurate recommendation. Note that when $K$ increases, the decline magnitude of KL Divergence is decreasing, since several noisy categories at the tail affect the analysis of consistency.

\subsection{Efficiency Analysis}
We conduct the efficiency analysis on ML1M dataset by comparing the training time of \ours and other methods. For DGCF, LightGCN, SimpleX and UltraGCN which need back propagation to train the model, we accumulate the training time until we obtain the optimal validation accuracy. For GF-CF, PGSP and \ours, we directly calculate its training time, including the time for SVD. 

Figure~\ref{fig:eff_ana} shows the experimental results. From the results, we can find that compared with the GSP-based CF methods, \ours is comparable to GF-CF and PGSP. Compared with the GCN-based methods, \ours is 8X faster than that of UltraGCN, which is the most efficient GCN-based method among the baselines. Therefore, we can conclude that \ours is highly efficient.

\subsection{Sensitivity Analysis}

We conduct the sensitivity analysis on ML1M dataset to study how two important hyper-parameters affect the performance of \ours. The experimental results are shown in Figure~\ref{fig:sen_ana}, and we only show the results of NDCG@10 and NDCG@20 for better presentation of their trends, however, other metrics can draw the same conclusion.
\subsubsection{Number of High Frequency Components $p_1$.} Figure~\ref{fig:sen_ana} (a) shows that the accuracy of \ours first increases and then decreases as $p_1$ increases. This is because as $p_1$ increases, more unique characteristics are captured by \ours, thereby improving the accuracy of interaction prediction. However, noisy interactions also corresponds to high-frequency components, and excessive high-frequency components will also adopt the noisy interactions to model user preference, thereby affecting the accuracy of interaction prediction. 

\subsubsection{Order of Item High-order Neighborhood Filter $k_1$.} Figure~\ref{fig:sen_ana} (b) shows as $k_1$ increases, the performance of the model has significantly improved, as more and more information from neighbors of different distances of item to be used to model user preference, thus improving the accuracy of recommendation results. However, an increase in $k_1$ will cause $(\mathbf{I}-\mathbf{O}_I)^{k_1}$ to become dense, thereby increasing the computational complexity of \ours. Simultaneously, the  performance gain of \ours is also gradually decreasing as $k_1$ becomes larger. Therefore, selecting an appropriate $k_1$ is necessary to balance the accuracy and running efficiency of the model. For the order of User High-order Neighborhood Filter $k_2$, we do not report it due to the space limitation but can draw the same conclusion.

\section{Related Work}
\subsection{GCN-based Recommendation Algorithm}
Recommender system, which studies the interactions between users and items, plays an important role in various domains, such as education~\cite{education1, education2, education3}, medical~\cite{medical1, medical2, medical3} and e-commerce~\cite{ecommerce1, ecommerce2}. Formally, user interactions can be constructed as a graph, where nodes are users and items, and edges are interactions between users and items. Due to the fact that graph convolutional networks (GCNs) has powerful structural feature extraction ability~\cite{gin, strucgnn}, more researchers begin to design recommendation algorithms based on GCNs, resulting in the rapid development of GCN-based recommendation algorithms~\cite{pinsage, dgcf, sgl, graphcfs1, graphcfs2, graphcfs3, graphcfs4}. GCMC~\cite{gcmc} is an early GCN-based recommendation algorithm that  introduced a graph auto-encoder to reconstruct user historical interactions, so as to predict user future interactions. PinSage~\cite{pinsage} combined random walks and graph convolutions to generate embeddings of items that incorporated both graph structure and node feature information, and designed a novel training strategy that relied on hard training examples to improve robustness and convergence of the model. UltraGCN~\cite{ultragcn} proposed a simple yet effective GCN-based recommendation algorithm which resorted to approximate the limit of infinite-layer graph convolutions via a constraint loss and allowed for more appropriate edge weight assignments and flexible adjustment of the relative importance among different relationships. 

It should be noted that the depth of GNN is restricted due to the over-smoothing problem~\cite{oversmoothing1, oversmoothing2}, where nodes tend to have similar representations. LR-GCCF~\cite{lrgccf} introduced the skip connection~\cite{skipconnection} to alleviate the over-smoothing problem and made GNN deeper. IMP-GCN~\cite{imp-gcn} proposed to propagate information in the user sub-graphs instead of the whole interaction graph, so as to reduce the impact of noise or negative information and alleviate the over-smoothing problem to make personalized recommendation.

\subsection{GSP-based Recommendation Algorithm}
GSP-based recommendation algorithms~\cite{cfgsp, gfcf, pgsp} are receiving increasing attention from researchers due to its parameter-free characteristic, making them highly efficient and able to achieve decent performance. GF-CF~\cite{gfcf} developed a unified graph convolution-based framework and proposed a simple yet effective collaborative filtering method which integrated a linear filter and an ideal low-pass filter to make recommendation. PGSP~\cite{pgsp} proposed a mixed-frequency low-pass filter over the personalized graph signal to model user preference and predict user interactions. However, these methods neglect to utilize user/item unique characteristics and user and item high-order neighborhood information to model user preference, making the modeled user preference sub-optimal.

\section{Conclusion}
We propose a frequency-aware graph signal processing method (\ours) for collaborative filtering to fully utilize information in user interactions for user future interaction prediction. \ours consists of a Cascaded Filter Module---which is composed of an ideal high-pass filter and an ideal low-pass filter---to take both user/item unique characteristics and common characteristics into consideration for user preference modeling, and a Parallel Filter Module---which is composed of two low-pass filters---to fully utilize user and item high-order neighborhood information for user preference modeling. Extensive experiments demonstrate the superiority of our method from the perspectives of recommendation accuracy and training efficiency compared to existing GCN/GSP-based CF methods.

\bibliographystyle{ACM-Reference-Format}
\bibliography{main}

\end{document}